\documentclass{llncs}

\newcommand{\Oh}[1]
    {\ensuremath{\mathcal{O}\!\left( {#1} \right)}}
\newcommand{\occ}
    {\ensuremath{\mathit{occ}}}

\begin{document}

\title{Indexes for Jumbled Pattern Matching\\in Strings, Trees and Graphs}
\author{Ferdinando Cicalese\inst{1}
    \and Travis Gagie\inst{2}\fnmsep\inst{3}
    \and Emanuele Giaquinta\inst{2}\fnmsep\thanks{Supported by Academy of Finland grant 118653 (ALGODAN).}
    \and Eduardo Sany Laber\inst{4}
    \and Zsuzsanna Lipt\'ak\inst{5}
    \and Romeo Rizzi\inst{5}
    \and Alexandru I.\ Tomescu\inst{2}\fnmsep\inst{3}\fnmsep\thanks{Supported by Academy of Finland grant 250345 (CoECGR).}}
\institute{Department of Computer Science, University of Salerno, Italy
    \and Department of Computer Science, University of Helsinki, Finland
    \and Helsinki Institute for Information Technology HIIT, Finland
    \and Department of Computer Science, PUC Rio de Janeiro, Brazil
    \and Department of Computer Science, University of Verona, Italy}
\maketitle

\begin{abstract}
We consider how to index strings, trees and graphs for jumbled pattern matching when we are asked to return a match if one exists.  For example, we show how, given a tree containing two colours, we can build a quadratic-space index with which we can find a match in time proportional to the size of the match.  We also show how we need only linear space if we are content with approximate matches.
\end{abstract}

\section{Introduction} \label{sec:introduction}

Suppose we are given a connected graph $G$ on $n$ coloured nodes and a multiset $M$ of colours and asked to find a connected subgraph of $G$ whose nodes' colours are exactly those in $M$, if such a subgraph exists.
Even when $G$ is a tree there can be exponentially many such matching subgraphs.  When $G$ is a path, however,
there are $\Oh{n}$ matches and we can find them all in $\Oh{n}$ time~\cite{BEL04}.  When $G$ is a path containing a constant number of colours, in $\Oh{n^2}$ time we can build a \(o (n^2)\)-space index with which we can determine in \(o (n)\) time whether there is a match~\cite{KR13}.  When $G$ is a path containing only two colours, in $\Oh{n^2 / \log^2 n}$ time we can build an $\Oh{n}$-bit index with which we can determine in $\Oh{1}$ time whether there is a match~\cite{CFL09,BCFL12,MR12,GHLW13}.  It follows that in $\Oh{n^2 / \log^2 n}$ time we can build an index of size  $\Oh{n \log n}$-bits  with which we can find all the matches using $\Oh{|M|}$ worst-case time per match~\cite{GHLW13}.  We can build an approximation of this index in $\Oh{n^{1 + \epsilon}}$ time with the quality of the approximation depending on $\epsilon$~\cite{CLWY12}.  Throughout this paper our model is the word-RAM with \(\Omega (\log n)\)-bit words and we measure space in words unless stated otherwise.

Determining whether there is a match is NP-complete even when $G$ is a tree~\cite{LFS06} or when it contains only two colours, but takes polynomial time when $G$ both has bounded treewidth and contains only a constant number of colours~\cite{FFHV11}.  When $G$ contains only two colours there exists an $\Oh{n}$-bit index with which we can determine in $\Oh{1}$ time whether there is a match~\cite{GHLW13}.  Building this index is NP-hard in general but, since finding a match is self-reducible, takes polynomial time when $G$ has bounded treewidth and $\Oh{n^2 / \log^2 n}$ time when $G$ is a tree.  At the cost of increasing the space to $\Oh{n}$ words, this index can be generalized to return a subset of the nodes in the matches that is also a hitting set for all the matches, using $\Oh{\log n}$ time worst-case time per match.  In the worst case, however, this subset of nodes is of little use in finding even a single complete match.

We start by presenting some basic tradeoffs in Section~\ref{sec:basics}.  In Sections~\ref{sec:strings} to~\ref{sec:trees approx} we assume $G$ contains only two colours.  In Section~\ref{sec:strings} we consider the case when $G$ is a path --- i.e., a binary string --- and describe an $\Oh{n}$-space index with which we can find a match in $\Oh{\log n}$ time.  In Section~\ref{sec:trees exact} we consider the case when $G$ is a tree and, based on our index for binary strings, describe an $\Oh{n^2}$-space index with which we can find a match in $\Oh{|M|}$ time.  If we are concerned only with multisets of size at most $n^{1 / 2}$, then we can reduce the space bound to $\Oh{n}$.  In Section~\ref{sec:trees approx} we show that we can achieve the same space bound if we are content with approximate matches.  In the full version of this paper we will partially extend our results to graphs, by working on spanning trees.

\section{Basic Tradeoffs} \label{sec:basics}

Suppose $G$ is a graph containing a constant number $c$ of colours and we will be given $M$ as the vector of length $c$ whose components are the frequencies of the characters, which is called the Parikh vector for $M$.  Since there are \({m + c - 1 \choose c - 1} = \Oh{m^{c - 1}}\) possible multisets of size $m$ and it takes $\Oh{m}$ space to store pointers to a match for such a multiset, there exists an $\Oh{n^{c + 1}}$-space index with which we can find a match in $\Oh{|M|}$ time.  When $G$ has bounded treewidth we can build this index in polynomial time, and we can reduce the space bound to $\Oh{n}$ at the cost of increasing the query time to $|M|^{\mathcal{O} (1)}$.  To do the latter, we store $G$ itself and pre-compute and store pointers to matches only for multisets of size at most $n^{1 / (c + 1)}$.  Given a multiset $M$ with \(|M| > n^{1 / (c + 1)}\), we search $G$ in \(n^{\mathcal{O} (1)} = |M|^{\mathcal{O} (1)}\) time.

For any positive constant $\epsilon$, we can build an $\Oh{n \log^c n}$-space approximate index with which, if $M$ has an exact match, then in $\Oh{1}$ time we can find a substring whose Parikh vector differs from $M$'s by at most a factor of \(1 + \epsilon\) in each component.  (This index does not tell us whether $M$ has an exact match, however, since we may find such a substring even when it does not.)  Without loss of generality, assume we are concerned only with multisets in which each character appears at least once; we can reduce the general case to \(2^c = \Oh{1}\) instances of this one.  We store a $c$-dimensional grid with each side having length \(\lfloor \log_{1 + \epsilon} n \rfloor + 1\).  For each point \((x_0, \ldots, x_{c - 1})\) in this grid, we store pointers to the nodes in a connected subgraph whose Parikh vector is between \(\left( \rule{0ex}{2ex} (1 + \epsilon)^{x_0}, \ldots, (1 + \epsilon)^{x_{c - 1}} \right)\) and \(\left( \rule{0ex}{2ex} (1 + \epsilon)^{x_0 + 1}, \ldots, (1 + \
epsilon)^{x_{c - 1} + 1} \right)\).  This takes a total of $\Oh{n \log^c n}$ space.  Given the Parikh vector \((v_0, \ldots, v_{c - 1})\) of $M$, we return the subgraph stored for the point \(\left( \rule{0ex}{2ex} \lfloor \log_{1 + \epsilon} v_0 \rfloor, \ldots, \lfloor \log_{1 + \epsilon} v_{c - 1} \rfloor \right)\) in the grid, if that subgraph exists.  We summarize these basic tradeoffs in the following lemma:

\begin{lemma} \label{lem:basics}
When $G$ is a graph containing a constant number $c$ of colours there exists an $\Oh{n^{c + 1}}$-space index with which we can find a match in $\Oh{|M|}$ time.  For any positive constant $\epsilon$ there exists an $\Oh{n \log^c n}$-space index with which in $\Oh{|M|}$ time we can find an approximate match in which each colour's frequency is within a factor of \(1 + \epsilon\) of its frequency in $M$.  When $G$ has bounded treewidth we can build these indexes in polynomial time and, moreover, we can reduce the space of the exact index to $\Oh{n}$ at the cost of increasing the query time to $|M|^{\mathcal{O} (1)}$.
\end{lemma}

When $G$ is a path --- which we can think of as a string over an alphabet of $c$ characters --- we can improve these bounds.  Since $G$ contains $\Oh{n^2}$ substrings and we can specify any substring by its two endpoints, we can build an $\Oh{n^2}$-space index with which we can find a match in $\Oh{1}$ time.  Calculation shows we can reduce the space bound to $\Oh{n}$ at the cost of increasing the query time to $\Oh{|M|^c}$, and we can store an approximate index in $\Oh{\log^c n}$ space.  In Appendix~\ref{app:multicoloured} we show how in $\Oh{n^{1 + \epsilon}}$ expected time we can build an index with which we can find all $\occ$ matches of $M$ in $\Oh{|M|^{1 / \epsilon} + \occ}$ time.

As an aside, we note that we can extend our approximate indexes to support approximate scaled-then-permuted pattern matching (see~\cite{BEL04}).  To do this, for each point \((x_0, \ldots, x_{c - 1})\) in the grid for which there is no subgraph whose Parikh vector is between \(\left( \rule{0ex}{2ex} (1 + \epsilon)^{x_0}, \ldots, (1 + \epsilon)^{x_{c - 1}} \right)\) and \(\left( \rule{0ex}{2ex} (1 + \epsilon)^{x_0 + 1}, \ldots, (1 + \epsilon)^{x_{c - 1} + 1} \right)\), we store pointers to the nodes in a connected subgraph (if there is one) whose Parikh vector is a multiple of a one between \(\left( \rule{0ex}{2ex} (1 + \epsilon)^{x_0}, \ldots, (1 + \epsilon)^{x_{c - 1}} \right)\) and \(\left( \rule{0ex}{2ex} (1 + \epsilon)^{x_0 + 1}, \ldots, (1 + \epsilon)^{x_{c - 1} + 1} \right)\).  The query time is still proportional to the size of the match returned but that may now be larger than $|M|$.

\section{An Index for Binary Strings} \label{sec:strings}

Suppose $G$ is a binary string, i.e., \(G [1..n] \in \{0, 1\}^*\).  If there are $p$ copies of 1 in \(G [i..i + m - 1]\) and $r$ copies of 1 in \(G [k..k + m - 1]\), then for every value $q$ between $p$ and $r$ there is a position $j$ between $i$ and $k$ such that \(G [j..j + m - 1]\) contains $q$ copies of 1.  This observation was the basis for the index in~\cite{CFL09} and is the basis for ours as well.

We store an $\Oh{1}$-time rank data structure for $G$ and, for \(1 \leq m \leq n\), we store the endpoints of two substrings of length $m$ in $G$ with the most and with the fewest copies of 1.  This takes a total of $\Oh{n}$ space. Given a Parikh vector \((v_0, v_1)\), we look up the left endpoints $i$ and $j$ of the substrings of length \(v_0 + v_1\) in $G$ with the most and with the fewest copies of 1.  We set $i$ and $j$ as the initial endpoints for a binary search: at each step, we use two rank queries to find the number $q$ of 1s in \(G \left[ \left\lfloor \frac{i + j}{2} \right\rfloor..\left\lfloor \frac{i + j}{2} \right\rfloor + v_0 + v_1 - 1 \right]\); if \(q = v_1\) then we stop and report this substring by its endpoints; if \(q < v_1\) then we set \(i = \lfloor (i + j) / 2 \rfloor\) and continue; if \(q > v_1\) then we set \(j = \lfloor (i + j) / 2 \rfloor\) and continue.  This search takes a total of $\Oh{\log n}$ time.

\begin{theorem} \label{thm:strings}
When $G$ is a path containing only two colours, we can build an $\Oh{n}$-space index with which we can find a match in $\Oh{\log n}$ time.
\end{theorem}

\section{Exact Indexes for Trees with Two Colours} \label{sec:trees exact}

Suppose $G$ is a tree containing only two colours, black and white.  Gagie, Hermelin, Landau and Weimann~\cite{GHLW13} noted that the observation in Section~\ref{sec:strings} can be extended to connected graphs: if there are connected subgraphs $H_p$ and $H_r$ in $G$ with $m$ nodes each and $p$ and $r$ white nodes, respectively, then for every value $q$ between $p$ and $r$, there is a connected subgraph $H_q$ with $m$ nodes and $q$ white nodes.

To see why, notice that we can construct a sequence of connected subgraphs with $m$ nodes such that the sequence starts with $H_p$ and ends with $H_r$ and any consecutive pair of subgraphs in the sequence differ on two nodes.  To build this sequence, we find a path between $H_p$ and $H_r$.  We root $H_p$ and $H_r$, which are trees themselves, at the first and last nodes in the path (or at a shared node, if they are not disjoint).  One by one, we remove nodes bottom-up in $H_p$ and add nodes along the path; remove nodes nearest to $H_p$ in the path and add nodes further along the path; then remove nodes from the path and add nodes top-down in $H_r$.

Suppose $p$ and $r$ are the minimum and maximum numbers of white nodes in any connected subgraphs of size $m$, and we store a path consisting of the nodes in $H_p$ in bottom-up order, followed by the nodes in the path, followed by the nodes in $H_r$ in top-down order.  If we apply Theorem~\ref{thm:strings} to this path, then we obtain an $\Oh{n}$-space index with which, given the Parikh vector for a multiset $M$ with \(|M| = m\), we can find a match in the graph $G$ in $\Oh{\log n + |M|}$ time.  Notice that, if \(|M| < \log n\), then we can simply store an $\Oh{\log^2 n}$-space lookup table with which we can find a match in $\Oh{|M|}$ time.  Therefore, applying this construction for \(1 \leq m \leq n\), we obtain the following theorem:

\begin{theorem} \label{thm:trees exact}
When $G$ is a tree containing only two colours, we can build an $\Oh{n^2}$-space index with which we can find a match in $\Oh{|M|}$ time.
\end{theorem}

When \(m \approx n\), we need $\Oh{n}$ space to store subgraphs with the minimum and maximum numbers of white nodes and the path between them.  When \(m \ll n\), however, those subgraphs are small and most of the space is taken up by the path.  We now claim we can store $G$ such that we can support fast rank queries on paths; due to space constraint, we leave the proof to Appendix~\ref{app:proofs}.

\begin{lemma} \label{lem:decomposition}
We can store $G$ in $\Oh{n}$ space such that $q$ rank queries on the path between any two nodes take a total of $\Oh{\log n + q}$ time.
\end{lemma}

If we store $G$ with Lemma~\ref{lem:decomposition} and store subgraphs with the minimum and maximum numbers of white nodes only for \(1 \leq m \leq n^{1 / 2}\), then our index takes only $\Oh{n}$ space but supports queries only for \(|M| \leq n^{1 / 2}\).  When \(|M| > n^{1 / 2}\) we can use an algorithm by Gagie et al.\ to find a match in \(\Oh{|M| n} = \Oh{|M|^3}\) time.

\begin{corollary} \label{cor:trees exact}
When $G$ is a tree containing only two colours, we can build an $\Oh{n}$-space index with which we can find a match in $\Oh{|M|}$ time when \(|M| \leq n^{1 / 2}\) and in $\Oh{|M|^3}$ time otherwise.
\end{corollary}

\section{An Approximate Index for Trees with Two Colours} \label{sec:trees approx}

In this section we present our most technical result, which is how to store in $\Oh{n}$ space an approximate index for a tree containing only two colours.  Again, an approximate match is one whose Parikh vector differs from $M$'s by a factor of at most \(1 + \epsilon\) in each component.  In contrast, with Lemma~\ref{lem:basics} we would use $\Oh{n \log^2 n}$ space.  Without loss of generality, assume we are only concerned with multisets in which there are at least as many black nodes as white nodes; we can build a symmetric index for the other case.  Notice that in this case, if we can find a connected subgraph $H$ with the same size as the given multiset $M$ and in which the number of white nodes is within a factor of \(1 + \epsilon\) of the number in $M$, then the number of black nodes in $H$ is also within a factor of \(1 + \epsilon\) of the number in $M$.

Our main idea is to store an $\Oh{n}$-space data structure with which, given a size $m$, we can find two connected subgraphs with size $m$ that have approximately the minimum and maximum numbers of white nodes.  Suppose we store a subgraph with the minimum number of white nodes for each size that is a power of two and for each size such that the minimum number of white nodes is a factor of \(1 + \epsilon\) greater than the number in the preceding stored subgraph.  That is, we store a sequence of \(\lg n\) subgraphs with total size $\Oh{n}$ and a sequence of \(\log_{1 + \epsilon} n\) subgraphs with total size $\Oh{n \log n}$.  The latter sequence of subgraphs has total size $\Oh{n \log n}$ in the worst case because the minimum number of white nodes may stay low until we reach size nearly $n$ and then increase rapidly, causing us to store about \(\log_{1 + \epsilon} n\) subgraphs each of size nearly $n$.  However, we can store this sequence of subgraphs in a total of $\Oh{n}$ space using the following lemma,
which we prove in Appendix~\ref{app:proofs}.  Similarly, we also store a subgraph with the maximum number of white nodes for each size that is a power of two and for each size such that the maximum number of white nodes is a factor of \(1 + \epsilon\) greater than the number in the preceding stored subgraph; this also takes $\Oh{n}$ total space if we  store the subgraphs with the following lemma.

\begin{lemma} \label{lem:subgraphs}
We can store $G$ in $\Oh{n}$ space such that, if $G$ contains a connected subgraph of size $m$ with $w$ white nodes, then we can represent some such subgraph in $\Oh{w}$ space such that recovering this subgraph takes $\Oh{m}$ time.
\end{lemma}

If we are given a multiset $M$ such that we have subgraphs of size $|M|$ sampled, then we can proceed as in the proof of Theorem~\ref{thm:trees exact} and find an exact match if there is one.  If we do not have subgraphs of size $|M|$ sampled, then we use our sampled subgraphs to build subgraphs $H_{\min}$ and $H_{\max}$ of size $|M|$ with approximately minimum and maximum numbers of white nodes, then proceed almost as in the proof of Theorem~\ref{thm:trees exact}: if the number of white nodes $H_{\min}$ is larger but within a factor of \(1 + \epsilon\) of the number in $M$, then we return $H_{\min}$; if the number in $H_{\min}$ is more than a factor of \(1 + \epsilon\) larger than the number in $M$, then there is no exact match and we return nothing; if the number of white nodes $H_{\max}$ is smaller but within a factor of \(1 + \epsilon\) of the number in $M$, then we return $H_{\max}$; if the number in $H_{\max}$ is less than a factor of \(1 + \epsilon\) smaller than the number in $M$, then there is no
exact match and we return nothing; in all other cases, we proceed as in Theorem~\ref{thm:trees exact}.

To build $H_{\min}$ we take the next larger subgraph with a minimum number of white nodes and discard nodes until it has size $|M|$ while leaving it connected.  This next larger subgraph has size less than \(2 |M|\), because we sampled for every size that is a power of two; has at most \(1 + \epsilon\) times more white nodes than the subgraph of size $|M|$ with the minimum number of white nodes, because we sampled whenever the minimum number of white nodes increased by a factor of \(1 + \epsilon\); and is a tree, because it is a connected subgraph of a tree.  It follows that discarding nodes takes $\Oh{|M|}$ time and, since discarding nodes cannot increase the number of white nodes, $H_{\min}$ contains at most \(1 + \epsilon\) times the minimum number of white nodes.  To build $H_{\max}$ we take the next smaller subgraph with a maximum number of white nodes and add nodes until it has size $|M|$.  By symmetric arguments, this takes $\Oh{|M|}$ time and, since adding nodes cannot decrease the number of white
nodes, the maximum number of white nodes in a subgraph of size $|M|$ is at most \(1 + \epsilon\) times the number in $H_{\max}$.  Finding the path from $H_{\min}$ to $H_{\max}$ takes $\Oh{|M|}$ time using the representation from Lemma~\ref{lem:decomposition}.

\begin{theorem} \label{thm:trees approx}
When $G$ is a tree containing only two colours, for any positive constant $\epsilon$ we can build an $\Oh{n}$-space index with which in $\Oh{|M|}$ time we can find an approximate match in which each colour's frequency is within a factor of \(1 + \epsilon\) of its frequency in $M$.
\end{theorem}

\bibliographystyle{plain}
\bibliography{jumbled}

\appendix

\section{An Index for Strings over Constant-Size Alphabets} \label{app:multicoloured}

Suppose $G$ is a string over a constant-size alphabet and \(0 < \epsilon \leq 1\).  Then in $\Oh{n^{1 + \epsilon}}$ expected time we can build an index with which, given a multiset $M$ of characters, we can find all $\occ$ matches of $M$ in $\Oh{|M|^{1 / \epsilon} + occ}$ worst-case time.  To do this, we store $G$ itself and, for \(1 \leq m \leq n^\epsilon\), we make a pass over $G$ and store, for each multiset of size $m$ that has a match in $G$, a list of all the locations of that multiset's matches.  Notice the lists for multisets of size $m$ are disjoint and have total length \(n - m + 1\); therefore, with dynamic perfect hashing we use a total of $\Oh{n^{1 + \epsilon}}$ expected time and $\Oh{n^{1 + \epsilon}}$ space.  Given a multiset $M$ with \(|M| \leq n^\epsilon\), we return our pre-computed list of the locations of $M$ matches in $\Oh{|M| + \occ}$ time, or $\Oh{\occ}$ time if we are given $M$ as a Parikh vector.  Given a multiset $M$ with \(|M| > n^\epsilon\), we search $G$ in \(\Oh{n} = \Oh{|M|^{1
/ \epsilon}}\) time.

\section{Proofs of Lemmas~\ref{lem:decomposition} and~\ref{lem:subgraphs}} \label{app:proofs}

\setcounter{lemma}{1}

\begin{lemma}
We can store $G$ in $\Oh{n}$ space such that $q$ rank queries on the path between any two nodes take a total of $\Oh{\log n + q}$ time.
\end{lemma}

\begin{proof}
We compute the heavy-path decomposition~\cite{ST83} of $G$ and store $\Oh{1}$-time rank data structures for each of the heavy paths, which takes $\Oh{n}$ space.  The path between any two nodes $u$ and $v$ is a sequence of $\Oh{\log n}$ intervals of heavy paths.  Given $u$ and $v$, for each of these intervals we compute the number of white nodes in that interval and to either side of it in the heavy path; this takes a total of $\Oh{\log n}$ time and rank queries on heavy paths.  With this information we can perform any rank query on the path from $u$ to $v$ using a single rank query on a heavy path.
\qed
\end{proof}

\begin{lemma}
We can store $G$ in $\Oh{n}$ space such that, if $G$ contains a connected subgraph of size $m$ with $w$ white nodes, then we can represent some such subgraph in $\Oh{w}$ space such that recovering this subgraph takes $\Oh{m}$ time.
\end{lemma}

\begin{proof}
We store the adjacency lists for $G$'s nodes, with each list ordered such that black neighbours precede white neighbours.  With this representation, we can expand a subgraph by adding only black nodes as long as this is possible, using $\Oh{1}$ time per added node.

Let $H$ be a connected subgraph of size $m$ with $w$ white nodes.  We store pointers to the white nodes in $H$, which takes $\Oh{w}$ space.  Since $G$ is a tree, we can find the unique paths between these nodes in a total of $m$ time; notice these paths are contained in $H$ and consist of black nodes.  If the subgraph consisting of the white nodes and these paths has fewer than $m$ nodes, then we add black nodes until it has $m$ nodes, which takes a total of $\Oh{m}$ time.  It is possible to add enough black nodes without adding any white nodes because, e.g., we could add the remaining black nodes in $H$.
\qed
\end{proof}

\end{document}